\documentclass[conference]{IEEEtran}

\usepackage{algorithm}
\usepackage[noend]{algorithmic}
\usepackage{graphicx}%
\usepackage{bm}%
\usepackage{amsfonts}
\usepackage{amssymb}
\usepackage{times}
\usepackage{subfigure}
\usepackage{enumitem}
\usepackage{latexsym,bm,amsmath,amssymb} 
\usepackage{CJK}
\usepackage{hhline}
\usepackage{cite}

\usepackage{multirow} 
\usepackage{xcolor}
\usepackage{ntheorem}
\usepackage{epstopdf}
\newtheorem{lemma}{\bf{Lemma}}
\newtheorem*{proof}{\it{Proof:}}

\usepackage{geometry}
\IEEEoverridecommandlockouts\IEEEpubid{\makebox[\columnwidth]{ 979-8-3503-1090-0/23/\$31.00 ©2023 Crown \hfill} \hspace{\columnsep}\makebox[\columnwidth]{ }}
\begin{document}

\title{QoS Aware Transmit Beamforming for Secure Backscattering in Symbiotic Radio Systems}

\vspace{-2mm}
\author{\IEEEauthorblockN{Mingcheng Nie\IEEEauthorrefmark{1},
Deepak Mishra\IEEEauthorrefmark{1}, Azzam Al-nahari\IEEEauthorrefmark{2}, Jinhong Yuan\IEEEauthorrefmark{1}, and  Riku J$\ddot{\text{a}}$ntti\IEEEauthorrefmark{2}
}
\IEEEauthorblockA{\IEEEauthorrefmark{1}School of Electrical Engineering and Telecommunications, University of New South Wales, Sydney, NSW 2052, Australia\\
\IEEEauthorrefmark{2}Department of Communications and 
Networking, Aalto University, Espoo 02150, Finland\\
Emails: m.nie@student.unsw.edu.au, d.mishra@unsw.edu.au, azzam.al-nahari@aalto.fi, \\ j.yuan@unsw.edu.au, and riku.jantti@aalto.fi}
}


\maketitle

\begin{abstract}
This paper focuses on secure backscatter transmission in the presence of a passive multi-antenna eavesdropper through a symbiotic radio (SR) network. Specifically, a single-antenna backscatter device (BD) aims to transmit confidential information to a primary receiver (PR) by using a multi-antenna primary transmitter's (PT) signal, where the received symbols are jointly decoded at the PR. Our objective is to achieve confidential communications for BD while ensuring that the primary system's quality of service (QoS) requirements are met. We propose an alternating optimisation algorithm that maximises the achievable secrecy rate of BD by jointly optimising primary transmit beamforming and power sharing between information and artificial noise (AN) signals. Numerical results verify our analytical claims on the optimality of the proposed solution and the proposed methodology's underlying low complexity. Additionally, our simulations provide nontrivial design insights into the critical system parameters and quantify the achievable gains over the relevant benchmark schemes. 
\end{abstract}


%
\IEEEpeerreviewmaketitle 
\vspace{-1mm}
\section{Introduction}
\vspace{-1mm}
Backscatter communication may provide a viable solution for future energy-efficient and affordable Internet-of-things (IoT) devices, as recognised by developing technology experts~\cite{mishra2018optimizing}. Recently, a new technology called Symbiotic Radio (SR) has been proposed as a means to achieve spectrum-sharing efficiency and reliable communications for IoT transmissions. In SR, passive backscatter devices (BD) \cite{9786078} use the ambient backscatter (AmBC) scheme to ride over the received signals from the licensed transmitter~\cite{Symbiotic_radio1}. By sharing the same receiver with the primary link, BD transmissions can avoid interference, allowing for reliable transmissions through joint decoding of the primary and backscatter transmissions, unlike cognitive radio (CR)~\cite{Symbiotic_radio2}. However, due to the low-cost BDs that can be attached to every physical object and the spectrum-sharing nature, malicious attacks on the BD tags can lead to data interception and privacy breaches \cite{Secure_communication}. Therefore, securing backscatter communication systems is a critical design issue. It has been discovered that Physical Layer Security (PLS) provides simpler security algorithms compared to cryptographic schemes~\cite{yang2015safeguarding}. This is crucial considering the size, cost, and computation limitations.
\begin{figure}
    \centering
    \includegraphics[width=2.6in]{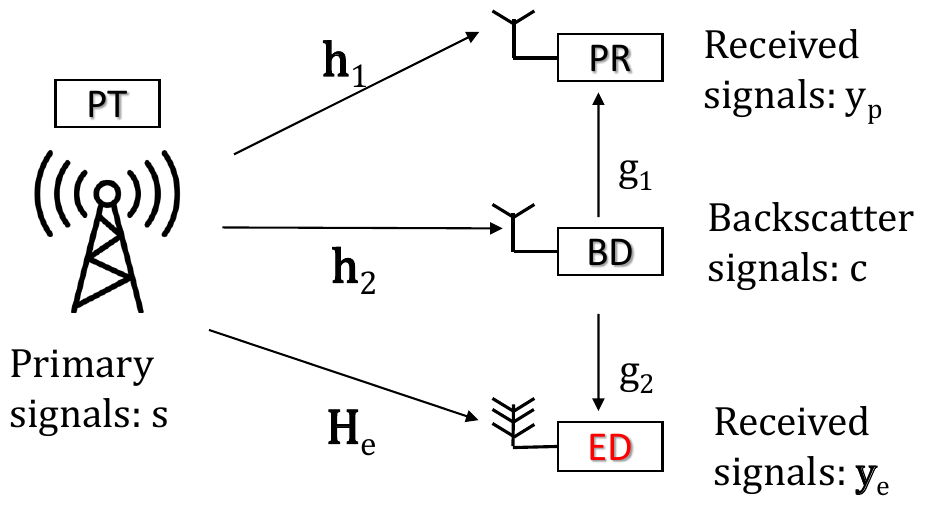}
    \vspace{-3mm} 
    \caption{SR model consists of PR, PT, BD and ED, where the BD backscatters its information to PR and the ED tries to decode the information of the BD.}
    \label{fig:system_model}
    \vspace{-2mm} 
\end{figure}

\vspace{-0.5mm}
\subsection{State-of-the-Art}
Studies on backscatter PLS can be categorized into two groups. The first group \cite{yang2020exploiting,shahzad2019covert,hassanieh2015securing} focuses on modifying the unmodulated carrier signal's properties through PLS to improve the decoding error rate for eavesdroppers. For instance, in \cite{yang2020exploiting}, authors used randomized continuous waves (CW) to achieve secure transmissions, where the secrecy rate is optimized by adjusting the CW's critical parameters. In contrast, \cite{shahzad2019covert} used a noise-like signal with varying power to enable covert backscatter communication. \cite{hassanieh2015securing} explored how randomized modulation and wireless channels can shield commercial RFID tags from eavesdropping when the reader lacked multiple antenna capacity, but the eavesdropper did not.

The second group \cite{saad2014physical,zhao2020safeguarding,yang2016physical} uses artificial noise (AN) or interference injection to decrease eavesdropper signal-to-noise-ratio (SNR). In \cite{saad2014physical}, AN signals are injected into conventional CW signals, with optimized power allocation between AN and CW signals. Similarly, \cite{zhao2020safeguarding} proposed AN-aided CW signals for secure backscatter transmission in the presence of proactive eavesdroppers. \cite{yang2016physical} investigated AN injection precoding strategy for secure MIMO backscatter communications, while \cite{saad2014physical,zhao2020safeguarding} considered single antennas and tags. Unlike the works above, \cite{li2019secure} suggested secure multiuser SR transmissions by incorporating non-orthogonal multiple access (NOMA) and optimizing corresponding beamforming vectors. \cite{li2021physical} conducted an outage and intercept probability analysis for a multiuser C-AmBC network, with single antennas considered at all nodes. Finally, \cite{li2023physical} proposed three physical layer authentication schemes for the AmBC-aided NOMA symbiotic network regarding the variations of authentication tags.

\vspace{-1mm}
\subsection{Motivation and Contributions}
\vspace{-1mm}
In this paper, compared to the existing works~\cite{yang2020exploiting,shahzad2019covert,hassanieh2015securing,zhao2020safeguarding,yang2016physical,saad2014physical,li2019secure,li2021physical,li2023physical}, we investigate the secure backscatter transmissions in multi-antenna SR systems by AN injection along with transmit beamforming. This work provides novel engineering design insights on optimal transmissions for secure SR networking in the presence of eavesdropping attackers. The main contributions are summarized next.
\begin{itemize}[leftmargin=*]
    \item We propose a secure transmission scheme for the multi-antenna SR system that takes into account quality of service (QoS) and employs AN injection. The scheme is designed to protect against a passive multi-antenna eavesdropper attempting to decode information sent by a passive single-antenna BD. We explore the proposed secure transmission scheme's performance bounds and robustness aspects. 
    \item We proposed to maximize the secure rate of backscatter communications by jointly optimizing the sources of the multi-antenna primary transmitter under the QoS requirements of the primary system. Specifically, we optimized the primary transmit beamforming and the power allocation between signal transmission and AN injection.
    \item Since the optimization is a non-convex problem, we developed a low-complexity alternating optimization algorithm with fast convergence speed. We also conducted a complexity analysis for this algorithm. Here, we have developed semi-closed form expressions for optimal solutions, which offer new insights for design. 
    \item  Numerical results verify our analytical claims regarding optimality and fast convergence with low complexity. We also provide optimal design insights on power allocation and beamforming vectors. Lastly, we conduct a performance comparison study where the proposed scheme is shown to outperform the relevant benchmark schemes.
\end{itemize}
$\mathbf{\textit{Notations:}}$ We define $[x]^+ \triangleq \max(0,x)$. Note that $\lvert \cdot \rvert$ and $\lVert \cdot \rVert$ are the absolute operation and Euclidean norm, respectively. We denote $\mathbf{h}^{T}$ and $\mathbf{h}^{\dag}$ as the transpose and complex conjugate transpose of $\mathbf{h}$, respectively. $\mathbf{I}_{N}$, $\mathbf{0}_{N}$, and $\mathbf{1}_{N}$ denote the $N\times N$ identity matrix, the all-zero column vector of length $N$, and the all-one column vector of length $N$, respectively. $\mathbf{x} \thicksim \mathcal{CN}(\mathbf{0}_{N},\pmb{\Sigma})$ indicates that $\mathbf{x} \in \mathbb{C}^{N\times 1}$ is the circularly symmetric complex Gaussian vector with zero-mean and covariance matrix $\pmb{\Sigma}$. Note that $\rm{v}_{max}\{\mathbf{M}\}$ represents the generalized principal eigenvector corresponding to maximum eigenvalue $\rm{\lambda}_{max}\{\mathbf{M}\}$ of matrix $\mathbf{M}$.

\section{System and Channel Models}\label{sec:System and Channel Models}
\subsection{SR Setup and Channel Model}
We consider an SR network, as shown in Fig.~\ref{fig:system_model}, which consists of a primary transmitter (PT) with $N_t$ antennas, a single-antenna backscatter device (BD), a single-antenna primary receiver (PR), and a multi-antenna eavesdropper (ED) with $N_e$ antennas. Note that $\mathbf{h}_1 \thicksim \mathcal{CN}(\mathbf{0}_{N_t},\sigma_s^2\mathbf{I}_{N_t})$ represents the channel fading vector from the PT to PR, $\mathbf{h}_2 \thicksim \mathcal{CN}(\mathbf{0}_{N_t},\sigma_c^2\mathbf{I}_{N_t})$ is the channel fading vector from the PT to BD, and $\mathbf{H}_e \thicksim \mathcal{CN}(\mathbf{0}_{N_tN_e},\sigma_e^2\mathbf{I}_{N_tN_e})$ is $N_e\times N_t$ channel fading matrix from the PT to ED. In this paper, we assume that $\mathbf{h}_1$ and $\mathbf{h}_2$ are available at the PT, which is commonplace in the literature~\cite{mishra2019optimal}, where the CSI can be obtained by channel reciprocity in time-division duplexing (TDD) systems.


\subsection{Transmission Signal Analysis}
The PT transmits a primary information symbol $s$ to the PR (primary link). Meanwhile, the BD transmits a secondary information signal $\sqrt{\alpha}c$ by riding over the PT signals (secondary link), where $\alpha$ denotes the reflection coefficient. We assume that the polyphase coding scheme is employed by the PT, i.e., $|s|^2=1$, and the Gaussian codebook is employed by the BD, i.e., 
$c\sim\mathcal{CN}(0,1)$. Note that this is commonplace for parasitic setup in~\cite{guo2019cooperative,kang2018riding}, where the target is to maximize the achievable rate of the secondary system. Moreover, the PT uses $N_t-2$ degrees of freedom for transmitting AN vector $\mathbf{z} = [z_1 z_2 \cdots z_{N_t-2}]^{T} \thicksim \mathcal{CN}(\mathbf{0}_{N_t-2},\mathbf{I}_{N_t-2})$. Thus, in this system model, $N_t > 2$. The transmitted signal at the PT is
\begin{align}
\label{trans_sig}
\mathbf{x} = \sqrt{p}\mathbf{w}s+\sqrt{q}\mathbf{W}\mathbf{z}=\sqrt{p}\mathbf{w}s+\sqrt{q}\sum_{i=1}^{N_t-2}\mathbf{w}_iz_i
\end{align}
where $p$ and $q$ are the transmitted power of the information and jamming signals, respectively, and $\mathbf{w}\in \mathbb{C}^{N_t\times1}$ is the normalized beamforming vector of the information signal, i.e., $\lVert \mathbf{w}\rVert = 1$. Moreover, the total transmitted power $P$ is constrained such that $\lVert \mathbf{x}\rVert^2 = P$. $\mathbf{W} = [\mathbf{w}_1 \mathbf{w}_2 \cdots \cdots \mathbf{w}_{N_t-2}] \in \mathbb{C}^{N_t\times(N_t-2)}$ is the precoding matrix of the jamming signal $\mathbf{z}$ with column normalization $\lVert \mathbf{w}_i\rVert = 1, \forall i$. In this paper, we design the AN to be completely suppressed at the PR, which leads to the precoding matrix $\mathbf{W}$ of the AN to lie in the null space of the channels $\mathbf{h}_1$ and $\mathbf{h}_2$, i.e., $\mathbf{h}_1^{\dag}\mathbf{W} = \mathbf{0}^T_{N_t-2}$ and $\mathbf{h}_2^{\dag}\mathbf{W} = \mathbf{0}^T_{N_t-2}$.  The null-space-based AN design will degrade the eavesdropping channels but not the legitimate channels to facilitate the secure transmission design~\cite{Beamforming_with}. We assume that $N_t > N_e$ because the eavesdropper cannot eliminate the AN term in (\ref{trans_sig}) with this condition \cite{secure_trans}. Considering $0 \le \phi \le 1$ denote the fraction of power devoted to the information signal, the transmitter powers $p$ and $q$ are given by
\begin{align}
    p &= \phi P,\label{p}\\
    q &= \frac{(1-\phi)P}{N_t-2}.\label{q}
\end{align}
Therefore, the received signal at the PR is given as
\begin{equation}\label{received-sig}
    y_p = \sqrt{p}\mathbf{h}_1^\dag\mathbf{w}s+\sqrt{p}\sqrt{\alpha}cg_1\mathbf{h}_2^\dag\mathbf{w}s+n_p,
\end{equation}
where the first term of the right-hand side in (\ref{received-sig}) is the received signals from the primary link, the second term is from the secondary link, $g_1\thicksim \mathcal{CN}(0,1)$ is the channel coefficient of the BD-PR link known at PT, and $n_p \thicksim \mathcal{CN}(0,1)$ represents the additive white Gaussian noise (AWGN) at PR. The received signals at the ED are given by
\begin{align}
    \mathbf{y}_e &=\sqrt{p}\mathbf{H}_e\mathbf{w}s+\sqrt{p}\sqrt{\alpha}cg_2\mathbf{1}_{N_e}\mathbf{h}_2^\dag\mathbf{w}s\nonumber\\
    &+\sqrt{q}\sum_{i=1}^{N_t-2}\mathbf{H}_e\mathbf{w}_i z_i+\mathbf{n}_e,\label{EDreceived}
\end{align}
where the first term of the right-hand side in (\ref{EDreceived}) is the received signals from the PT-ED link, the second term is from the PT-BD-ED link, $g_2\thicksim \mathcal{CN}(0,1)$ is the channel coefficient of the BD-ED link known at PT, and $\mathbf{n}_e \thicksim\mathcal{CN}(\mathbf{0}_{N_e},\mathbf{I}_{N_e})$. 

In this paper, we consider the worst-case scenario, where the ED has zero noise, i.e., $\mathbf{n}_e \to \mathbf{0}_{N_e}$, and the ED can decode the PT signal for the sake of intercepting the BD signal. This will result in an upper bound on the achievable rate of the ED and a lower bound on the secrecy rate~\cite{Beamforming_with}\cite{secure_trans}. Therefore, the received signal at the ED receiver is given as
\begin{equation}\label{eve}
    \Tilde{\mathbf{y}}_e = \sqrt{p}\sqrt{\alpha}cg_2\mathbf{1}_{N_e}\mathbf{h}_2^\dag\mathbf{w}s+\sqrt{q}\sum_{i=1}^{N_t-2}\mathbf{H}_e\mathbf{w}_i z_i.
\end{equation}

\section{Problem Definition}\label{sec:Problem Definition}
\subsection{SNR Analysis and Secrecy Rate Definition}
Here we derive the achievable secrecy rate of the considered system setup. We are considering the parasitic case where the symbol period of the BD is equal to that of the primary system~\cite{Symbiotic_radio1}. So, the BD signal is treated as interference, and the SNR of the primary system is given from (\ref{received-sig}) as
\begin{equation}
    \begin{aligned}
    \label{gamma_s}
    \gamma_s &=\frac{p\lvert \mathbf{h}_1^\dag \mathbf{w} \rvert^2}{p\alpha \lvert g_1 \rvert^2 \lvert \mathbf{h}_2^\dag \mathbf{w} \rvert^2 +1}.
    \end{aligned}
\end{equation}
After decoding the primary link signal $s$ and removing it from the received signal in (\ref{received-sig}) by successive interference cancellation (SIC) technique, the SNR of the BD signal is
\begin{equation}
\label{gamma_c}
    \gamma_{c|s} =  p\alpha\lvert g_1\rvert^2\lvert\mathbf{h}_2^\dag \mathbf{w}\rvert^2.
\end{equation}
We assume the eavesdropper to be aware of $\mathbf{H}_e\mathbf{w},\mathbf{h}_2^\dag\mathbf{w}$, and the correlation matrix $q\mathbf{H}_e\mathbf{W}\mathbf{W}^\dag\mathbf{H}_e^\dag$ of the AN signal to perform the optimal detection that maximizes its SNR $\gamma_{e|s}$~\cite{Robust_beamforming}. Here we define $\mathbf{X} \triangleq \mathbf{H}_e\mathbf{W}\mathbf{W}^\dag\mathbf{H}_e^\dag$ and the SNR at the ED is given as
\begin{equation}
\begin{aligned}
\label{gamma_e}
\gamma_{e|s} = \frac{p\alpha \lvert g_2\rvert^2}{q}\mathbf{w}^\dag\mathbf{h}_2\mathbf{1}_{N_e}^\dag\mathbf{X}^{-1}\mathbf{1}_{N_e}\mathbf{h}_2^\dag\mathbf{w}.
\end{aligned}
\end{equation}
The instantaneous achievable secrecy rate $R_{sec}$ is defined by
\begin{equation}
\begin{aligned}
\label{rate}
R_{sec} = \left[R_c - R_e \right]^+,
\end{aligned}
\end{equation}
where $R_c = \log_2(1+\gamma_{c|s})$ and $R_e = \log_2(1+\gamma_{e|s})$ represent the achievable rates at the backscatter and eavesdropper side, respectively. Here, we expand $R_{sec}$ for later use as follow
\begin{align}
    R_{sec}&= \Big[\log_2\Big(1+p \alpha \lvert g_1 \rvert^2 \lvert \mathbf{h}_2^\dag \mathbf{w} \rvert^2 \Big)- \nonumber\\
    &\ \ \log_2\Big(1+\frac{p \alpha \lvert g_2 \rvert^2}{q}\mathbf{w}^\dag \mathbf{h}_2 \mathbf{1}_{N_e}^\dag \mathbf{X}^{-1} \mathbf{1}_{N_e} \mathbf{h}_2^\dag \mathbf{w}\Big)\Big]^+.\label{R2}
\end{align}

\subsection{Problem Definition of Secrecy Rate Optimization}
Our goal is to maximize the achievable secrecy rate in~(\ref{rate}) in terms of power allocation factor and beamforming vector, subject to the transmitting power and QoS constraints. Thus, the optimization problem is formulated as follows
\begin{equation}
\begin{aligned}
\nonumber
\mathcal{O}_{1}: &\max_{\mathbf{w},\phi} R_{sec} = \left[R_c - R_e\right]^+,\ \  \text{subject to:}\\
&(\mathrm{C1}): \lVert\mathbf{w}\rVert^2 \le 1, \ \
(\mathrm{C2}): 0 \le \phi \le 1,\ \ (\mathrm{C3}):\gamma_s \ge \gamma_s^{th},
\end{aligned}
\end{equation}
where $\gamma_s^{th}$ is the minimum QoS requirement for PT in terms of SNR. Note that the constraint $(\mathrm{C1})$ is convex~\cite{mishra2019sum,mishra2019multi}, $(\mathrm{C2})$ is linear, and $(\mathrm{C3})$ is linear with $\phi$ as $\frac{\partial\gamma_s}{\partial\phi}=\frac{P\lvert\mathbf{h}_1^\dag \mathbf{w} \rvert^2}{(p\alpha \lvert g_1\rvert^2\lvert \mathbf{h}_2^\dag \mathbf{w}\rvert^2+1)^2} > 0$. However, $(\mathrm{C3})$ is nonconvex due to $\frac{\partial^2\gamma_s}{\partial\mathbf{w}^2}<0$ and the coupling between $\mathbf{w}$ and $\phi$ in $\gamma_s$ in~(\ref{gamma_s}). Thus, $\mathcal{O}_{1}$ is a nonconvex problem because both constraint $(\mathrm{C3})$ and $R_{sec}$ include the coupling terms between $\mathbf{w}$ and $\phi$~\cite{mishra2019sum}.

 
\section{Proposed Secrecy Rate Optimization}\label{sec:Full CSI solution}
\vspace{-1mm}
Here we propose the optimal solution for the problem $\mathcal{O}_{1}$ by alternately optimizing $\phi$ and $\mathbf{w}$. Specifically, we investigate the optimal information beamforming vector $\mathbf{w}$ for a given power allocation factor $\phi$ and optimal $\phi$ for a given $\mathbf{w}$ in the following two subsections. In this way, the optimal $\mathbf{w}$ and $\phi$ can be obtained through alternating and iterative updates. 
\subsection{Optimal $\mathbf{w}$ for a Given $\phi$}\label{sec:full CSI optimal w for given phi}
The problem of optimal $\mathbf{w}$ that maximizes the achievable secrecy rate for a given $\phi$ can be defined as
\begin{align}
    \mathcal{O}_{1.1}: &\max_{\mathbf{w}} R_{sec},\ \  \text{subject to:} \ (\mathrm{C1}), (\mathrm{C3}).\nonumber
\end{align}
\subsubsection{Feasible analysis}\label{sec:feasible_full}
 Before investigating the optimal solution of $\mathcal{O}_{1.1}$, we discuss the feasibility condition of $(\mathrm{C3})$ by finding the maximum achievable SNR $\gamma_s^{\max}$ at the PR. We start by rewriting the $\gamma_s$ in (\ref{gamma_s}) and $(\mathrm{C3})$ in simplified form.
\begin{lemma}
     $\gamma_s$ can be rewritten and simplified as
\end{lemma}
\begin{align}\label{simplified1}
    \gamma_{s} =\frac{\mathbf{w}^\dag \mathbf{G}_1 \mathbf{w}} {\mathbf{w}^\dag \mathbf{G}_2 \mathbf{w}},
\end{align}
where $\mathbf{G}_1 = \phi P \mathbf{h}_1 \mathbf{h}_1^\dag$ and $\mathbf{G}_2 =\alpha \phi P \lvert  g_1 \rvert^2 \mathbf{h}_2 \mathbf{h}_2^\dag +\mathbf{I}_{N_t}$ are both symmetry matrix.
\begin{proof}
Note that $\lvert \mathbf{h}_1^\dag \mathbf{w} \rvert^2=\mathbf{h}_1^\dag\mathbf{w}\mathbf{w}^\dag\mathbf{h}_1=\mathbf{w}^\dag\mathbf{h}_1\mathbf{h}_1^\dag\mathbf{w}$. Thus, $\gamma_{s}$ can be written as $\gamma_{s}=\frac{ \phi P \mathbf{w}^\dag \mathbf{h}_1 \mathbf{h}_1^\dag \mathbf{w}} {\mathbf{w}^\dag( \alpha \phi P \lvert  g_1 \rvert^2 \mathbf{h}_2 \mathbf{h}_2^\dag+\mathbf{I}_{N_t})\mathbf{w}}=\frac{\mathbf{w}^\dag \mathbf{G}_1 \mathbf{w}} {\mathbf{w}^\dag \mathbf{G}_2 \mathbf{w}}.\hfill\blacksquare$
\end{proof} 

We can observe that $\gamma_s$ in (\ref{simplified1}) is a generalized Rayleigh quotient. Thus, the global optimal beamforming vector $\mathbf{w}_{{e_1}}$ that maximizes $\gamma_{s}$ can be obtained by the generalized principal eigenvector of the matrix set $(\mathbf{G}_1,\mathbf{G}_2)$ as~\cite{saini2022irs}
\begin{equation}\label{w_e1}
\begin{aligned}
\mathbf{w}_{{e_1}}=\rm{v}_{max}\{(\mathbf{G}_1,\mathbf{G}_2)\}.
\end{aligned}
\end{equation}
The maximum SNR $\gamma_{s}^{\max}$ for given $\phi$ can be obtained by substituting $\mathbf{w}_{{e_1}}$ in (\ref{simplified1}). Thus, $\mathcal{O}_{1.1}$ is feasible if $\gamma_{s}^{max} \ge \gamma_s^{th}$.

\subsubsection{Proposed Optimal Solution of $\mathbf{w}$}
In order to solve $\mathcal{O}_{1.1}$, we rewrite and simplify the achievable secrecy rate $R_{sec}$ in (\ref{R2}) by treating $\phi$ as a constant and write in terms of $\mathbf{w}$ as
\begin{align}\label{simplified2}
R_{sec}&= \Big[\log_2\Big(1+p \alpha \lvert g_1 \rvert^2 \mathbf{w}^\dag \mathbf{h}_2 \mathbf{h}_2^\dag \mathbf{w} \Big) \nonumber\\
&\ \ -\log_2\Big(1+\frac{p \alpha \lvert g_2 \rvert^2}{q}\mathbf{w}^\dag \mathbf{h}_2 \mathbf{1}_{N_e}^\dag \mathbf{X}^{-1} \mathbf{1}_{N_e} \mathbf{h}_2^\dag \mathbf{w}\Big)\Big]^+,\nonumber\\
&= \Bigg[\log_2\Bigg(   \mathbf{w}^\dag\Big(\mathbf{I}_{N_t}+p \alpha \lvert g_1 \rvert^2  \mathbf{h}_2 \mathbf{h}_2^\dag \Big)\mathbf{w} \Bigg) \nonumber\\
&\ \ -\log_2\Bigg( \mathbf{w}^\dag\Big(\mathbf{I}_{N_t}+\frac{p \alpha \lvert g_2 \rvert^2}{q} \mathbf{h}_2 \mathbf{1}_{N_e}^\dag \mathbf{X}^{-1} \mathbf{1}_{N_e} \mathbf{h}_2^\dag \Big)\mathbf{w}\Bigg)\Bigg]^+,\nonumber\\
&=\left[\log_2 \left( \frac{\mathbf{w}^\dag\mathbf{G}_{3}\mathbf{w}}
{\mathbf{w}^\dag\mathbf{G}_{4}\mathbf{w}}\right)\right]^+,
\end{align}
where $\mathbf{G}_{3}=\mathbf{I}_{N_t}+p \alpha \lvert g_1 \rvert^2  \mathbf{h}_2 \mathbf{h}_2^\dag$ and 
$\mathbf{G}_{4}=\mathbf{I}_{N_t}+\frac{p \alpha \lvert g_2 \rvert^2}{q} \mathbf{h}_2 \mathbf{1}_{N_e}^\dag \mathbf{X}^{-1} \mathbf{1}_{N_e} \mathbf{h}_2^\dag$. This $R_{sec}$ in~(\ref{simplified2}) is a generalized Rayleigh quotient. Thus, the optimal beamforming vector $\mathbf{w}_{e_2}$ that maximizes $R_{sec}$ in (\ref{simplified2}) \emph{without constraints} is the generalized principal eigenvector of matrix set $(\mathbf{G}_3,\mathbf{G}_4)$ as~\cite{saini2022irs}
\begin{equation}\label{w_e2}
\begin{aligned}
\mathbf{w}_{e_2}=\rm{v}_{max}\{(\mathbf{G}_3,\mathbf{G}_4)\}.
\end{aligned}
\end{equation}

After investigating the optimal beamforming vector that maximizes the achievable secrecy rate without considering QoS constraint, it is crucial to strike a balance between the secrecy rate maximization and QoS requirement. Thus, we propose a weighted combination of $\mathbf{w}_{e_1}$ and $\mathbf{w}_{e_2}$ as follow
\begin{equation}
\begin{aligned}\label{wgt1}
\mathbf{w}_{c_1} = \frac{\lambda_1 \mathbf{w}_{e_1} + (1-\lambda_1)\mathbf{w}_{e_2}}{\lVert \lambda_1 \mathbf{w}_{e_1} + (1-\lambda_1)\mathbf{w}_{e_2} \rVert},
\end{aligned}
\end{equation}
where $\lambda_1$ is the weighting factor and varies in $d$ discrete steps uniformly, resulting in the allocation as $\{ 0,\frac{1}{d},\frac{2}{d}, \cdots , \frac{d-1}{d},1 \}$. Note that $\mathbf{w}_{e_1}$ in~(\ref{w_e1}) and $\mathbf{w}_{e_2}$ in~(\ref{w_e2}) are two extremes that maximize the received SNR at PR and unconstrained secrecy rate, respectively, and the optimal $\mathbf{w}_{c_1}$ in~(\ref{wgt1}) balances between those extremes. Here $d$ is chosen based on the tradeoff between the computational complexity and the desired solution quality. To compute the optimal $\mathbf{w}_{c_1}$, we need to evaluate $R_{sec}$ and $\gamma_s$ for all $\lambda$ weights and then choose the maximum constrained secrecy rate among them.

\subsection{Optimal $\phi$ for a Given $\mathbf{w}$}\label{sec:full CSI optimal phi for given w}
For a given $\mathbf{w}$, the problem of optimal $\phi$ that maximizes the achievable secrecy rate, subject to total power constraint $(\mathrm{C2})$, can be defined as
\begin{align}
    \mathcal{O}_{1.2}: &\max_{\phi} R_{sec},\ \  \text{subject to:} \ (\mathrm{C2}), (\mathrm{C3}).\nonumber
\end{align}
Note that we assume $\phi\neq 0$ because the PT has to send the information signals, and $\phi\neq 1$ because the eavesdropper will obtain an infinity rate, and the secure rate will be 0. Thus, the power factor is chosen as $0 < \phi < 1$. To obtain the solution of $\mathcal{O}_{1.2}$, we next rewrite $R_{sec}$ in (\ref{R2}) by treating $\mathbf{w}$ as a constant.
\begin{lemma}  
With $A\triangleq P\alpha \lvert g_1 \rvert^2 \lvert \mathbf{h}_2^\dag \mathbf{w} \rvert^2$, optimal $\phi$ is given as 
\begin{align}\label{full CSI optimal phi}
    \phi &= \frac{A-\sqrt{AB(A-B+1)}}{A-AB}.
\end{align}
where $B\triangleq (N_t-2)\alpha \lvert {g}_2 \rvert^2 \mathbf{w}^\dag \mathbf{h}_2 \mathbf{1}_{N_e}^\dag \mathbf{X}^{-1} \mathbf{1}_{N_e} \mathbf{h}_2^\dag \mathbf{w}$.
\end{lemma}
\begin{proof}
Firstly,  it is worth noting that $R_{sec}$ can be rewritten in terms of $\phi$ as
\begin{align}\label{phi}
    R_{sec}=\left[\log_2 \frac{\left(1+\phi A\right)}{\left(1+\frac{\phi}{(1-\phi)}B\right)}\right]^+.
\end{align}
Here, the obtained optimal $\phi$ is infeasible if it does not fall within the range $(0,1)$.
    To obtain a positive secrecy rate, we need $1+\phi A > 1+\frac{\phi}{(1-\phi)}B \Rightarrow (1-\phi)A>B$, which leads to $A>B$ as $0<\phi<1$. Note that $R_{sec}$ has two critical points with respect to $\phi$ by taking $\frac{\partial R_{sec}}{\partial\phi} = 0$. The critical points are shown as follows
\begin{align}
\phi_1 &= \frac{A-\sqrt{AB(A-B+1)}}{A-AB},\label{phi_1}\\
\phi_2 &= \frac{A+\sqrt{AB(A-B+1)}}{A-AB},\label{phi_2}
\end{align}
where $A>0$ and $B>0$. Then we take the second-order derivative as $\frac{\partial^2 R_{sec}}{\partial\phi^2} =\frac{2B\left(B-A-1\right)}{\left(\left(B-1\right)\phi+1\right)^3}$,
where its numerator is negative as $A>B$ and its denominator is positive as $(B-1)\phi > -1, \forall B>0$. Thus, $R_{sec}$ is concave and has two maximum values as $\frac{\partial^2 R_{sec}}{\partial\phi^2} < 0$. To select the feasible one among $\phi_1$ and $\phi_2$, we first analyse the case when $A-AB<0 \Rightarrow B>1$. In this case, $\phi_2$ in~(\ref{phi_2}) will always be negative as its numerator is positive. When $A-AB>0 \Rightarrow B<1$, we analysis $\phi_2$ in~(\ref{phi_2}) as follow
\begin{align}
    0<\phi_2<1\Rightarrow& 0<A+\sqrt{AB(A-B+1)}<A-AB,\nonumber\\
    \Rightarrow& -A<\sqrt{AB(A-B+1)}<-AB,\nonumber
\end{align}
where $\sqrt{AB(A-B+1)}<-AB$ is impossible. Thus, $\phi_2$ in~(\ref{phi_2}) is infeasible and only $\phi_1$ in~(\ref{phi_1}) is feasibile.
$\hfill\blacksquare$.
\end{proof}
\subsection{Step-by-Step Algorithm}
Next, we show the step-by-step procedure in Algorithm~\ref{algo1}. Specifically, Algorithm {\ref{algo1}} starts with a given power factor $\phi = 0.5$. Then, we obtain the optimal $\mathbf{w}$ for a given $\phi$ as shown in Section~\ref{sec:full CSI optimal w for given phi}, including SNR feasibility check and computation of the weighting beamforming vector $\mathbf{w}_{c_1}$. Based on the obtained optimal $\mathbf{w}$, we update the power factor as shown in Section~\ref{sec:full CSI optimal phi for given w} and in Algorithm {\ref{algo1}} line 22. Finally, Algorithm {\ref{algo1}} terminates when $(R_{sec}^{(j-1)} - R_{sec}^{(j-2)}) \le \varepsilon$, where $\varepsilon$ is an acceptable tolerance.

\begin{algorithm}
    \caption{\footnotesize Alternating optimization of $\mathbf{w}$ and $\phi$ to maximize $R_{sec}$}
    \label{algo1}
    \begin{algorithmic}[1]\footnotesize
        \REQUIRE $\mathbf{h}_1$, $\mathbf{h}_2$, $\mathbf{H}_e$, $\mathbf{X}$, $P$, $\alpha$, $g_1$, $g_2$, $d$, $\gamma_s^{th}$, $\varepsilon$
        \STATE Set $j\gets 1$, $\phi^{(1)} \gets 0.5$, $R_{sec}^{(1)} \gets 0$
        \REPEAT
        \STATE Obtain $\mathbf{w}_{{e_1}}$ by substituting $\phi\gets\phi^{(j)}$ into (\ref{simplified1}) and (\ref{w_e1})
        \STATE Obtain $\gamma_{s}^{\max}$ by substituting $\mathbf{w}\gets \mathbf{w}_{{e_1}}$ into (\ref{simplified1})
        \IF{$\gamma_s^{th} > \gamma_{s}^{\max}$}
        \PRINT $\mathcal{O}_1$ is not feasible
        \RETURN
        \ELSE
        \STATE Obtain $\mathbf{w}_{e_2}$ by substituting $\phi \gets \phi^{(j)}$ into (\ref{simplified2}), (\ref{w_e2})
        \STATE Set $i \gets 0$
        \FOR{$i \le d$}
        \STATE Set $\lambda_1 = \frac{i}{d}$ in (\ref{wgt1}) and set the resultant as $\mathbf{w}_{c_1}$
        \STATE Substitute $\mathbf{w}_{c_1}$ and $\phi^{(j)}$ into $\gamma_{s}$ in (\ref{simplified1}) and set the resultant as $\gamma_{temp}$
        \STATE Substitute $\mathbf{w}_{c_1}$ and $\phi^{(j)}$ into $R_{sec}^{(j)}$ in (\ref{simplified2}) and set the resultant as $R_{temp}$
        \IF{$\gamma_{temp} \ge \gamma_{s}^{th}$ and $R_{temp}^{(j)} > R_{sec}$}
        \STATE Set $R_{sec}^{(j)} \gets R_{temp}, \mathbf{w}_{opt} \gets \mathbf{w}_{c_1}$
        \ENDIF
        \STATE Set $i\gets i + 1$
        \ENDFOR
        \STATE Set $j \gets j + 1$
        \STATE Substitute $\mathbf{w}\gets\mathbf{w}_{opt}$ into (\ref{phi}) and (\ref{full CSI optimal phi}) and set the resultant as $\phi^{(j)}$
        \ENDIF
        \UNTIL $(R_{sec}^{(j-1)} - R_{sec}^{(j-2)}) \le \varepsilon$
        \STATE $R_{sec}\gets R_{sec}^{(j-1)}$, $\phi_{opt}\gets\phi^{(j-1)}$
        \ENSURE $\phi_{opt}$, $\mathbf{w}_{opt}$, $R_{sec}$
    \end{algorithmic}
\end{algorithm}

\subsection{Complexity Analysis}\label{sec: complexity of full CSI}
We first consider the computational complexity of Section~\ref{sec:full CSI optimal w for given phi}. It is worth noting that the main computational complexity comes from the generalized eigenvectors of the matrix set $(\mathbf{G}_1,\mathbf{G}_2)$ and $(\mathbf{G}_3,\mathbf{G}_4)$ in~(\ref{w_e1}) and~(\ref{w_e2}), respectively. Specifically, the generalized eigenvector problem of two symmetric matrices $(\mathbf{G}_1,\mathbf{G}_2)$ is given as~\cite{parlett1998symmetric}
\begin{align}
    \mathbf{G}_1\mathbf{V}=\mathbf{G}_2\mathbf{V}\mathbf{D},
\end{align}
where $\mathbf{V}$ and $\mathbf{D}$ contain the eigenvectors and eigenvalues, respectively. According to~\cite{faber1997solving}, this problem is solved in MATLAB based on matrix inversion as 
\begin{align}
    \mathbf{G}_2^{-1}\mathbf{G}_1\mathbf{V}=\mathbf{V}\mathbf{D}.
\end{align}
Therefore, the complexity of generalized eigenvectors of two symmetric matrices problem consists of channel inversion and normal eigenvalue computing. Generally, the computational complexity of matrix inversion is $\mathcal{O}\left(N_t^3\right)$~\cite{li2021cross} and normal eigenvalue computing is also $\mathcal{O}\left(N_t^3\right)$~\cite{pan1999complexity}. Thus, the complexity of finding a generalized eigenvector is $\mathcal{O}\left(2N_t^3\right)$. The same approach applied to the matrix set $(\mathbf{G}_3,\mathbf{G}_4)$.

Moreover, the complexity of iterations of finding $\mathbf{w}_{c_1}$ in~(\ref{wgt1}) is $\mathcal{O}\left(d+1\right)$. Note that the complexity of Section~\ref{sec:full CSI optimal phi for given w} is $\mathcal{O}\left(1\right)$. The complexity $\mathcal{O}\left(J\right)$ is due to the iterations of convergence, where $J$ is the iteration number of convergence. Finally, we summarize the computational complexity of Algorithm {\ref{algo1}} as $\mathcal{O}\left(J\left(4N_t^3(d+1)+1\right)\right)$.

\section{Numerical Results}\label{sec:Numerical Results}
Unless otherwise stated, we set $P=48\mathrm{dBm}$, $\gamma_s^{th}=3\mathrm{dB}$, $\gamma_c^{th}=10\mathrm{dB}$, $\alpha=0.3$, $N_t=10$, $N_e=4$, $\varepsilon=10^{-10}$, and $d=100$. We assume $\sigma^2_s=\sigma^2_c=\sigma^2_e=1$. Note that the MATLAB seed is set as $\mathrm{rng}(5)$, and the simulation results are averaged over $10^4$ times channel realization.
\begin{figure}
	\centering  
	\subfigure[Validation of $\phi$ optimization.]{{\includegraphics[scale=0.244]{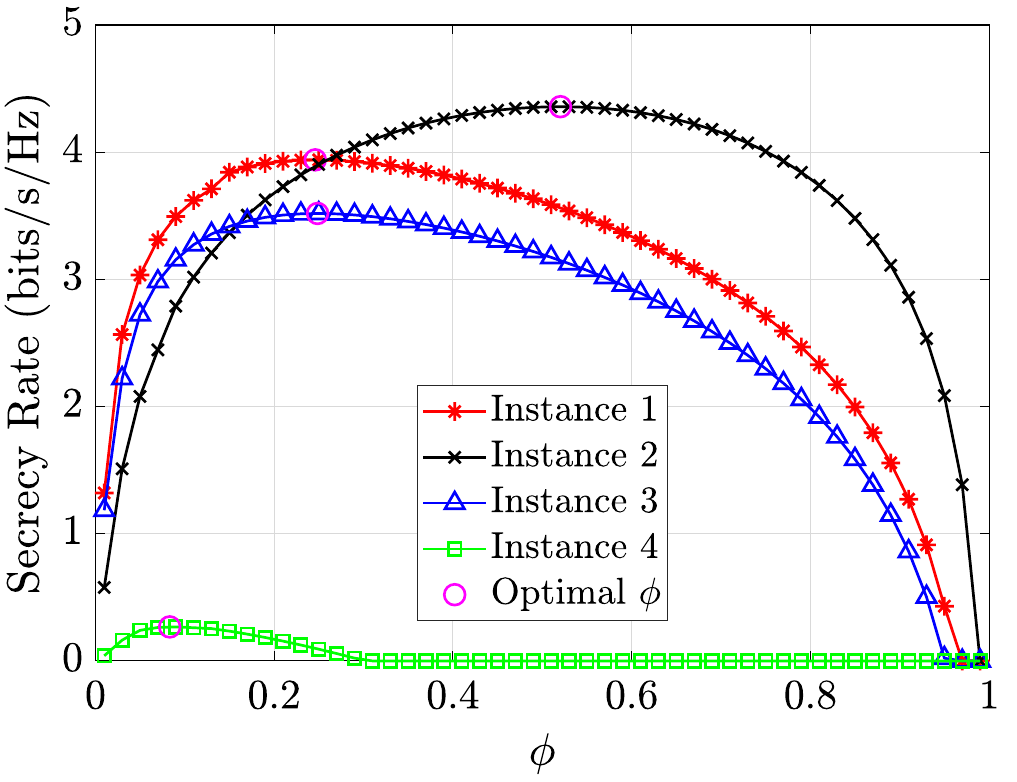} }}
	\subfigure[Validation of $\lambda_1$ optimization.]{{\includegraphics[scale=0.244]{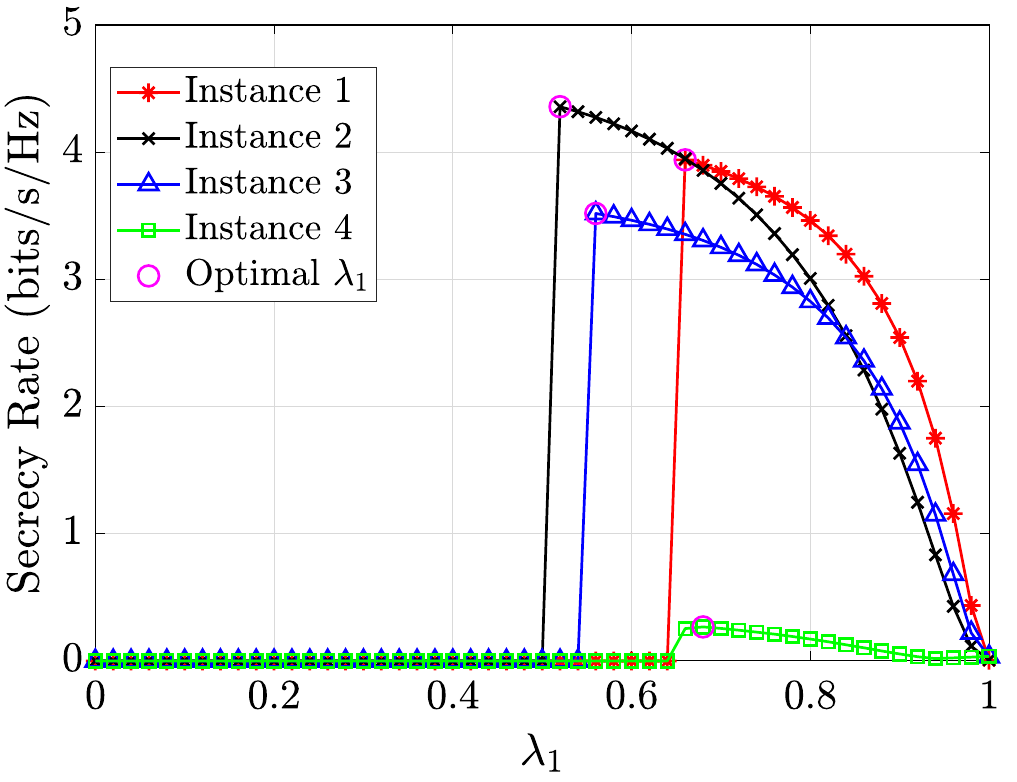}}}
    \caption{Validation of the proposed secrecy rate maximization algorithm} 
        \label{fig:group}
\end{figure}  

\begin{figure}
	\centering   
	\subfigure[Convergence demonstration with increasing number of iterations.]{\includegraphics[scale=0.244]{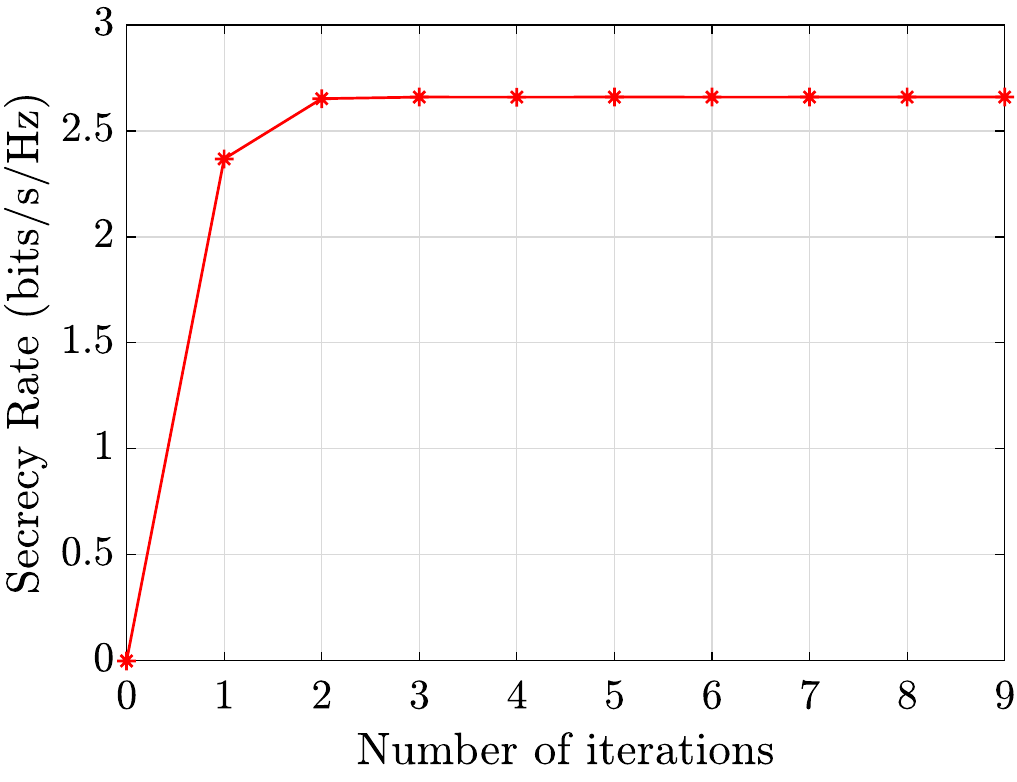}}
	\subfigure[Convergence demonstration with more sensitive tolerance $\varepsilon$.]{{\includegraphics[scale=0.244]{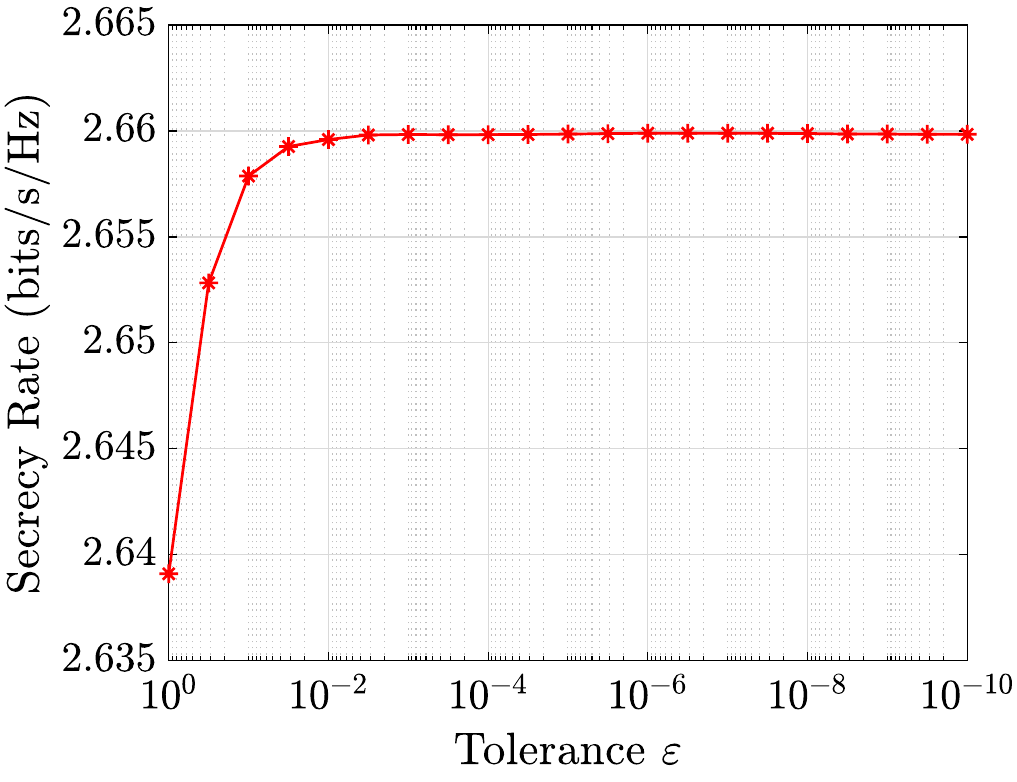}}}
    \caption{Convergence of the proposed secrecy rate maximization algorithm} 
        \label{fig:convergence}
\end{figure}  
\begin{figure}[!h]
    \centering
    \includegraphics[width=1.85in]{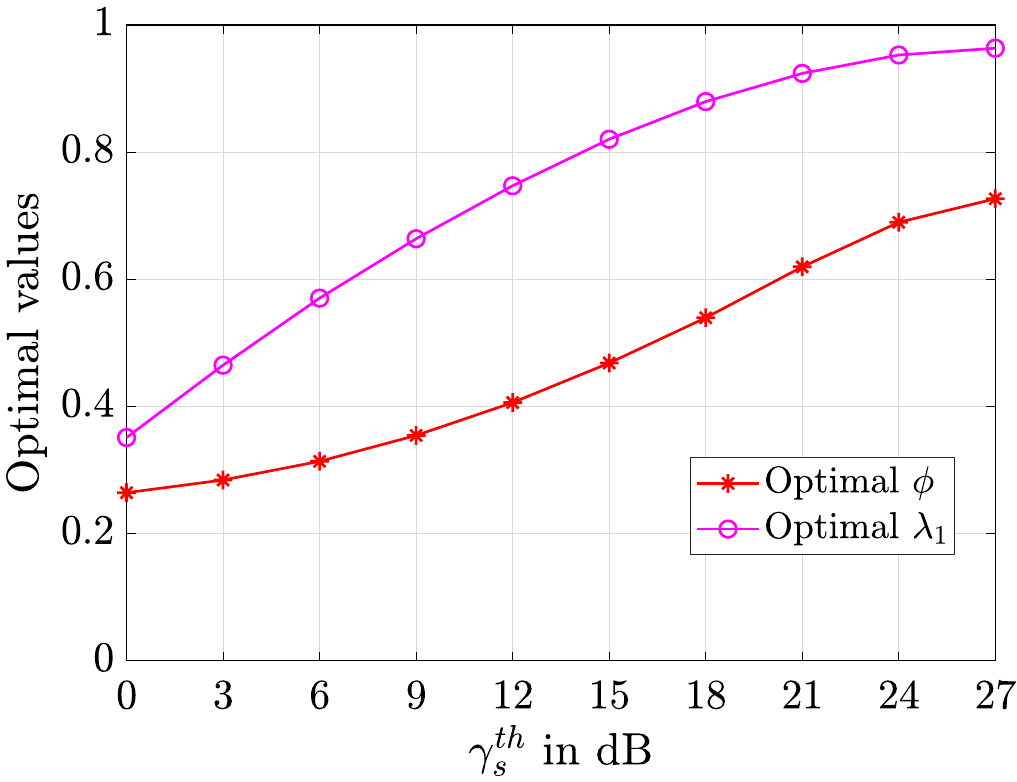}
    \caption{Optimal $\phi$ and $\lambda_1$ versus $\gamma_s^{th}$}
    \label{fig:insight phi lambda}
\end{figure}
First, we present the validation plots for the proposed algorithm in Fig.~\ref{fig:group}. Specifically, Fig.~\ref{fig:group}(a) and Fig.~\ref{fig:group}(b) show the secrecy rate against $\phi$ and $\lambda_1$ for four individual instances, respectively, where the optimal values for $\phi$ and $\lambda_1$ obtained by algorithm \ref{algo1} that maximize the secrecy rate  are denoted by magenta circles. We can observe that the globally optimal values match well with the search simulation results for all instances. Note that the secrecy rate is almost zero when $\lambda_1$ is approximately less than 0.5 in Fig.~\ref{fig:group}(b). This is because we define the secrecy rate to be a positive value in~$\eqref{rate}$, where the eavesdropper rate can be larger than BD's rate when the eavesdropper's channel $\mathbf{H}_e$ is much stronger than the primary link's channel conditions. It is important to acknowledge that our algorithm works on individual optimality rather than joint and ergodic optimality over group realizations, i.e., the sum of individual optimality may not be equal to the joint optimality of the group. Furthermore, the proposed algorithm \ref{algo1} will converge fast within four iterations and tolerance $10^{-4}$ as shown in Fig.~\ref{fig:convergence}(a) and Fig.~\ref{fig:convergence}(b), respectively. 

\begin{figure}
    \centering
    \includegraphics[width=1.85in]{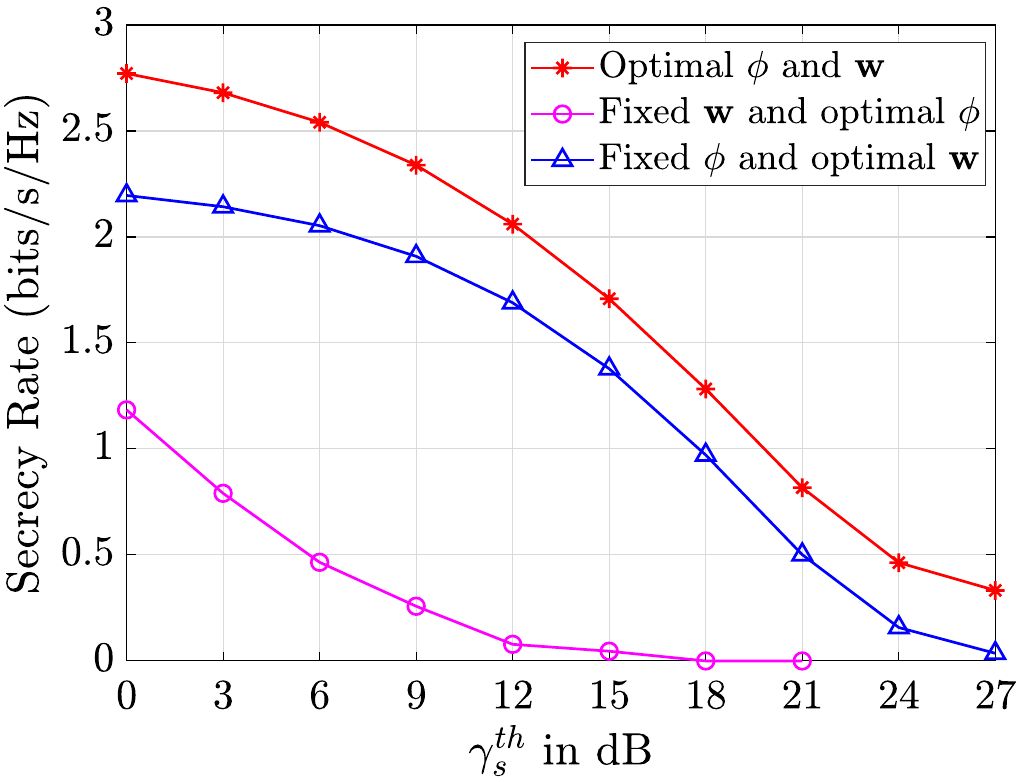}\vspace{-4mm}
    \caption{Secrecy rate comparison between proposed and benchmark schemes, where the fixed beamforming vector is MRT, and the fixed $\phi$ is $0.5$.}
    \label{fig:rate comparison}
\end{figure}
Next, insightful plots are presented for key algorithm parameters. Fig.~\ref{fig:insight phi lambda} shows that the values of optimal $\phi$ and $\lambda_1$ tend to increase as the threshold $\gamma_s^{th}$ increases. This can be attributed to the fact that a more significant value of $\phi$ results in higher transmission power in the primary link rather than jamming an eavesdropper. In contrast, a more significant value of $\lambda_1$ directs the beamforming vector closer to $\mathbf{w}_{e_1}$, thereby maximizing the received SNR at PR. Consequently, we can expect a decrease in the achievable secrecy rate with an increase in the threshold $\gamma_s^{th}$, as shown in Fig.~\ref{fig:rate comparison}. Furthermore, in Fig.~\ref{fig:rate comparison}, the proposed scheme outperforms the conventional schemes at different QoS requirements based on achievable rate comparison. Note that the red line denotes the secrecy rate based on the proposed algorithm, where both power allocation factor $\phi$ and beamforming vector $\mathbf{w}$ are optimized, the blue line denotes the optimal $\phi$ with maximum ratio transmitting (MRT) vector $\mathbf{w}_{\rm{MRT}} = \frac{\mathbf{h}_2}{\lVert \mathbf{h}_2 \rVert}$, and the magenta line denotes the optimal $\mathbf{w}$ with a fixed $\phi=0.5$. Moreover, it can be shown from Fig.~\ref{fig:insight}(a) that the secrecy rate is directly proportional to the number of antennas at PT $N_t$ and inversely proportional to the number of antennas at Eve $N_e$. Fig.~\ref{fig:insight}(b) demonstrates that increased transmission power $P$ and reflection coefficient $\alpha$ result in a higher secrecy rate.

\begin{figure}
	\centering  \vspace{-1mm}
	\subfigure[Variation with $N_t$ and $N_e$.]{{\includegraphics[scale=0.244]{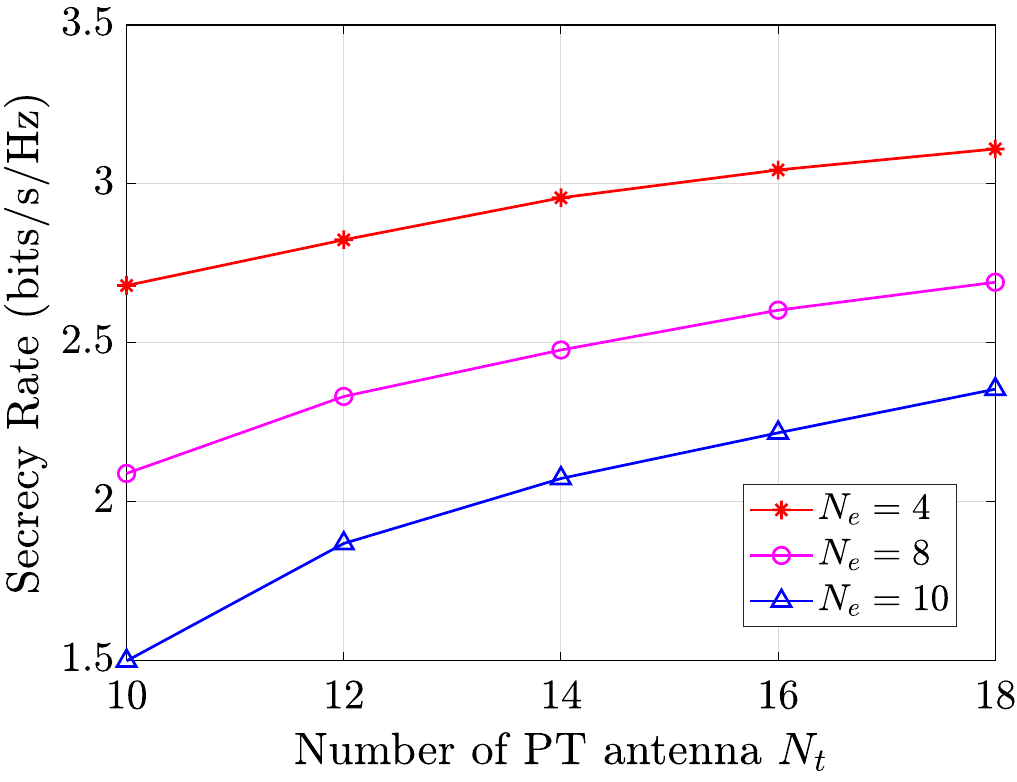}}}
	\subfigure[Variation with $P$ and $\alpha$.]{{\includegraphics[scale=0.244]{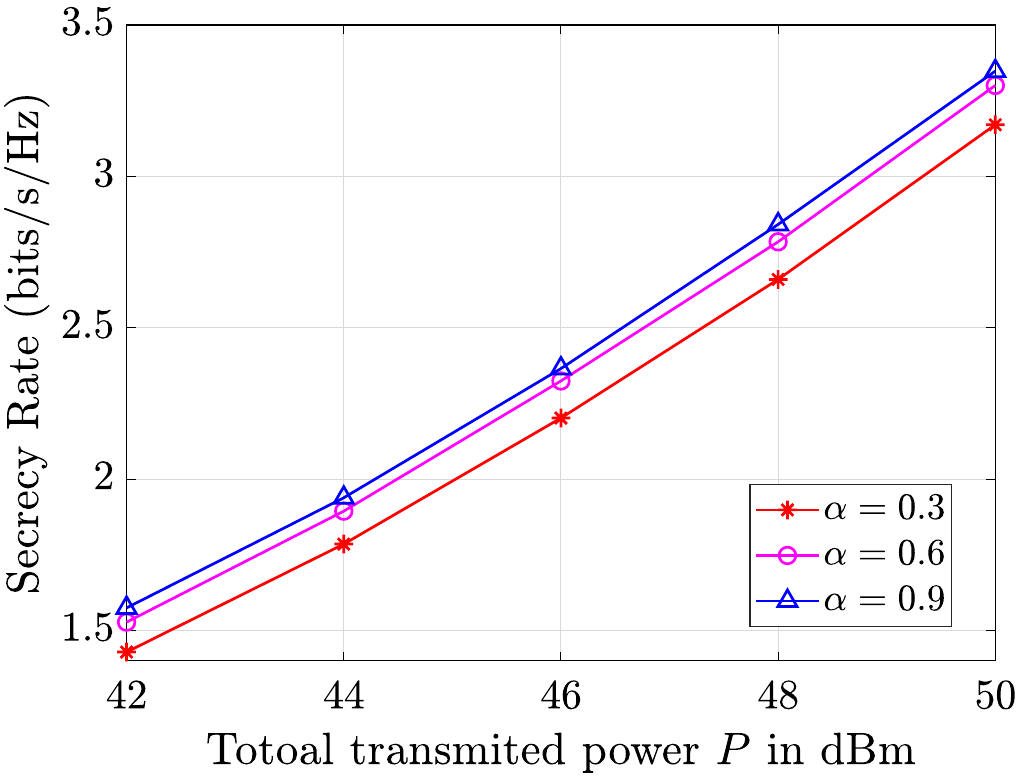}}}\vspace{-3mm}
    \caption{Insights on achievable secrecy rate $R_{sec}$ for different values of the number of antennas, total transmitted power $P$, and reflection coefficient $\alpha$.} 
        \label{fig:insight}\vspace{-3mm}
\end{figure}

\section{Conclusion}\label{sec:Conclusion}\vspace{-1mm}
In this paper, we delve into the topic of secure transmissions for SR networks with multiple antennas using AN injection. Firstly, we introduce the setup of the secure transmission system. Then, we devise an alternating optimization algorithm to maximize the secrecy rate by designing the power allocation factor $\phi$ and beamforming vector $\mathbf{w}$. Our findings reveal that the achievable secrecy rate is significantly impacted by the QoS constraints of the primary system and BD. Furthermore, our secure SR system designs can be extended to multiple tags, primary receivers with multiple antennas, and colluding eavesdroppers in the future.

\vspace{-0.2mm}
\section{Acknowledgment}\vspace{-1mm}
This research work has been supported in part by the Australian Research Council Discovery Early Career Researcher Award (DECRA) - DE230101391.
\vspace{-1mm}
\bibliographystyle{IEEEtran}
\bibliography{ref}

\begin{thebibliography}{10}
\providecommand{\url}[1]{#1}
\csname url@samestyle\endcsname
\providecommand{\newblock}{\relax}
\providecommand{\bibinfo}[2]{#2}
\providecommand{\BIBentrySTDinterwordspacing}{\spaceskip=0pt\relax}
\providecommand{\BIBentryALTinterwordstretchfactor}{4}
\providecommand{\BIBentryALTinterwordspacing}{\spaceskip=\fontdimen2\font plus
\BIBentryALTinterwordstretchfactor\fontdimen3\font minus
  \fontdimen4\font\relax}
\providecommand{\BIBforeignlanguage}[2]{{%
\expandafter\ifx\csname l@#1\endcsname\relax
\typeout{** WARNING: IEEEtran.bst: No hyphenation pattern has been}%
\typeout{** loaded for the language `#1'. Using the pattern for}%
\typeout{** the default language instead.}%
\else
\language=\csname l@#1\endcsname
\fi
#2}}
\providecommand{\BIBdecl}{\relax}
\BIBdecl

\bibitem{mishra2018optimizing}
D.~Mishra and E.~G. Larsson, ``Optimizing reciprocity-based backscattering with
  a full-duplex antenna array reader,'' in \emph{Proc. IEEE Int. Workshop
  Signal Process. Adv. Wireless Commun.}, Kalamata, Greece, Jun. 2018.

\bibitem{9786078}
A.~C.~Y. Goay, D.~Mishra, and A.~Seneviratne, ``{ASK} modulator design for
  passive {RFID} tags in backscatter communication systems,'' in \emph{Proc.
  IEEE WAMICON}, 2022, pp. 1--4.

\bibitem{Symbiotic_radio1}
R.~Long, Y.-C. Liang, H.~Guo, G.~Yang, and R.~Zhang, ``Symbiotic radio: A new
  communication paradigm for passive internet of things,'' \emph{IEEE Internet
  Things J}, vol.~7, no.~2, pp. 1350--1363, Feb. 2020.

\bibitem{Symbiotic_radio2}
Y.-C. Liang, Q.~Zhang, E.~G. Larsson, and G.~Y. Li, ``Symbiotic radio:
  Cognitive backscattering communications for future wireless networks,''
  \emph{IEEE Trans. Cognit. Commun. Netw.}, vol.~6, no.~4, pp. 1242--1255, Dec.
  2020.

\bibitem{Secure_communication}
Y.~Zhang, F.~Gao, L.~Fan, X.~Lei, and G.~K. Karagiannidis, ``Secure
  communications for multi-tag backscatter systems,'' \emph{IEEE Wireless
  Commun. Lett.}, vol.~8, no.~4, pp. 1146--1149, Aug. 2019.

\bibitem{yang2015safeguarding}
N.~Yang, L.~Wang, G.~Geraci, M.~Elkashlan, J.~Yuan, and M.~Di~Renzo,
  ``Safeguarding {5G} wireless communication networks using physical layer
  security,'' \emph{IEEE Commun. Mag.}, vol.~53, no.~4, pp. 20--27, Apr. 2015.

\bibitem{yang2020exploiting}
Q.~Yang, H.-M. Wang, Q.~Yin, and A.~L. Swindlehurst, ``Exploiting randomized
  continuous wave in secure backscatter communications,'' \emph{IEEE Internet
  Things J}, vol.~7, no.~4, pp. 3389--3403, Apr. 2020.

\bibitem{shahzad2019covert}
K.~Shahzad and X.~Zhou, ``Covert communication in backscatter radio,'' in
  \emph{Proc. IEEE ICC}, May. 2019, pp. 1--6.

\bibitem{hassanieh2015securing}
H.~Hassanieh, J.~Wang, D.~Katabi, and T.~Kohno, ``Securing {RFIDs} by
  randomizing the modulation and channel,'' in \emph{Proc. 12th USENIX Symp.
  Netw. Syst. Design Implement.}, 2015, pp. 235--249.

\bibitem{saad2014physical}
W.~Saad, X.~Zhou, Z.~Han, and H.~V. Poor, ``On the physical layer security of
  backscatter wireless systems,'' \emph{IEEE Trans. Wireless Commun.,},
  vol.~13, no.~6, pp. 3442--3451, Jun. 2014.

\bibitem{zhao2020safeguarding}
B.-Q. Zhao, H.-M. Wang, and P.~Liu, ``Safeguarding {RFID} wireless
  communication against proactive eavesdropping,'' \emph{IEEE Internet Things
  J.}, vol.~7, no.~12, pp. 11\,587--11\,600, Dec. 2020.

\bibitem{yang2016physical}
Q.~Yang, H.-M. Wang, Y.~Zhang, and Z.~Han, ``Physical layer security in {MIMO}
  backscatter wireless systems,'' \emph{IEEE Trans. Wireless Commun.}, vol.~15,
  no.~11, pp. 7547--7560, Nov. 2016.

\bibitem{li2019secure}
Y.~Li, M.~Jiang, Q.~Zhang, and J.~Qin, ``Secure beamforming in {MISO} {NOMA}
  backscatter device aided symbiotic radio networks,'' \emph{arXiv preprint
  arXiv:1906.03410}, 2019.

\bibitem{li2021physical}
X.~Li, Y.~Zheng, W.~U. Khan, M.~Zeng, D.~Li, G.~Ragesh, and L.~Li, ``Physical
  layer security of cognitive ambient backscatter communications for green
  {Internet-of-Things},'' \emph{IEEE Trans. Green Commun.}, vol.~5, no.~3, pp.
  1066--1076, Sep. 2021.

\bibitem{li2023physical}
X.~Li, Q.~Wang, M.~Zeng, Y.~Liu, S.~Dang, T.~A. Tsiftsis, and O.~A. Dobre,
  ``Physical-layer authentication for ambient backscatter-aided {NOMA}
  symbiotic systems,'' \emph{IEEE Trans. Commun.}, Feb. 2023.

\bibitem{mishra2019optimal}
D.~Mishra and E.~G. Larsson, ``Optimal channel estimation for reciprocity-based
  backscattering with a full-duplex {MIMO} reader,'' \emph{IEEE Trans. Signal
  Process.}, vol.~67, no.~6, pp. 1662--1677, Mar. 2019.

\bibitem{guo2019cooperative}
H.~Guo, Y.-C. Liang, R.~Long, and Q.~Zhang, ``Cooperative ambient backscatter
  system: A symbiotic radio paradigm for passive {IoT},'' \emph{IEEE Wireless
  Commun. Lett.}, vol.~8, no.~4, pp. 1191--1194, Aug. 2019.

\bibitem{kang2018riding}
X.~Kang, Y.-C. Liang, and J.~Yang, ``Riding on the primary: A new spectrum
  sharing paradigm for wireless-powered {IoT} devices,'' \emph{IEEE Trans.
  Wireless Commun}, vol.~17, no.~9, pp. 6335--6347, Sep. 2018.

\bibitem{Beamforming_with}
A.~Al-Nahari, G.~Geraci, M.~Al-Jamali, M.~H. Ahmed, and N.~Yang, ``Beamforming
  with artificial noise for secure {MISOME} cognitive radio transmissions,''
  \emph{IEEE Trans. Inf. Forensics Security}, vol.~13, no.~8, pp. 1875--1889,
  Aug. 2018.

\bibitem{secure_trans}
X.~Zhou and M.~R. McKay, ``Secure transmission with artificial noise over
  fading channels: Achievable rate and optimal power allocation,'' \emph{IEEE
  Trans. Veh. Technol.}, vol.~59, no.~8, pp. 3831--3842, Oct. 2010.

\bibitem{Robust_beamforming}
A.~Mukherjee and A.~L. Swindlehurst, ``Robust beamforming for security in
  {MIMO} wiretap channels with imperfect {CSI},'' \emph{IEEE Trans. Signal
  Process.}, vol.~59, no.~1, pp. 351--361, Jan. 2011.

\bibitem{mishra2019sum}
D.~Mishra and E.~G. Larsson, ``Sum throughput maximization in multi-tag
  backscattering to multiantenna reader,'' \emph{IEEE Trans. Commun.}, vol.~67,
  no.~8, pp. 5689--5705, Aug. 2019.

\bibitem{mishra2019multi}
------, ``Multi-tag backscattering to {MIMO} reader: Channel estimation and
  throughput fairness,'' \emph{IEEE Trans. Wireless Commun.}, vol.~18, no.~12,
  pp. 5584--5599, Dec. 2019.

\bibitem{saini2022irs}
R.~Saini, D.~Mishra, W.~Xiong, and J.~Yuan, ``{IRS-Assisted} secure {OFDMA}
  with untrusted users,'' in \emph{Proc. IEEE Global Commun. Conf. Workshops
  (GC Wkshps)}, Dec. 2022, pp. 619--624.

\bibitem{parlett1998symmetric}
B.~N. Parlett, \emph{{The Symmetric Eigenvalue Problem}}.\hskip 1em plus 0.5em
  minus 0.4em\relax Englewood Cliffs, NJ, USA: Prentice-Hall, 1998.

\bibitem{faber1997solving}
K.~Faber, ``On solving generalized eigenvalue problems using {MATLAB},''
  \emph{J. Chemom}, vol.~11, no.~1, pp. 87--91, Jun. 1997.

\bibitem{li2021cross}
S.~Li, W.~Yuan, Z.~Wei, and J.~Yuan, ``Cross domain iterative detection for
  orthogonal time frequency space modulation,'' \emph{IEEE Trans. Wireless
  Commun.}, vol.~21, no.~4, pp. 2227--2242, Sept. 2021.

\bibitem{pan1999complexity}
V.~Y. Pan and Z.~Q. Chen, ``The complexity of the matrix eigenproblem,'' in
  \emph{Proc. ACM Symp. Theory Comput.}, May. 1999, pp. 507--516.

\end{thebibliography}
\end{document}